\documentclass[letterpaper]{article} 
\usepackage[draft]{aaai2026}  
\usepackage{times}  
\usepackage{helvet}  
\usepackage{courier}  
\usepackage[hyphens]{url}  
\usepackage{graphicx} 
\usepackage{natbib}  
\usepackage{caption} 
\usepackage{algorithm}
\usepackage{amsthm}
\usepackage{thm-restate}

\usepackage{newfloat}
\usepackage{listings}
\DeclareCaptionStyle{ruled}{labelfont=normalfont,labelsep=colon,strut=off} 
\floatstyle{ruled}
\newfloat{listing}{tb}{lst}{}
\floatname{listing}{Listing}
\usepackage{amsmath, logicproof, stmaryrd, amssymb}
\usepackage{tikz}
\usetikzlibrary{shapes.misc, fit, decorations.pathreplacing,calligraphy,positioning} 
\usepackage{float}
\usepackage{algpseudocode}
\usepackage{algorithm}

\theoremstyle{definition}
\newtheorem{definition}{Definition}

\newtheorem{theorem}{Theorem}
\newtheorem{proposition}{Proposition}
\newtheorem{example}{Example}

\newcommand{\Upd}{\mathit{Upd}}
\newcommand{\Pay}{\mathit{Pay}}
\newcommand{\Nfin}{N^{\mathit{fin}}}
\newcommand{\Mfin}{M_{\mathit{fin}}}

\algtext*{EndIf}
\algtext*{EndFor}
\algtext*{EndProcedure}

\algnewcommand\algorithmiccase{\textbf{case}}
\algdef{SE}[CASE]{Case}{EndCase}[1]{\algorithmiccase\ #1}{\algorithmicend\ \algorithmiccase}%
\algtext*{EndCase}
\algrenewcommand\algorithmicindent{1em} 

\makeatletter
\newenvironment{breakablealgorithm}
  {
   \begin{center}
     \refstepcounter{algorithm}
     \hrule height.8pt depth0pt \kern2pt
     \renewcommand{\caption}[2][\relax]{
       {\raggedright\textbf{\fname@algorithm~\thealgorithm} ##2\par}%
       \ifx\relax##1\relax 
         \addcontentsline{loa}{algorithm}{\protect\numberline{\thealgorithm}##2}%
       \else 
         \addcontentsline{loa}{algorithm}{\protect\numberline{\thealgorithm}##1}%
       \fi
       \kern2pt\hrule\kern2pt
     }
  }{
     \kern2pt\hrule\relax
   \end{center}
  }
\makeatother

\allowdisplaybreaks

\title{Formal Verification of Diffusion Auctions\thanks{This is an extended version of the paper with the same title that will appear in the proceedings of AAAI 2026. This version contains a technical appendix with proof details that, for space reasons, do not appear in the AAAI 2026 version.}
}
\author{ 
   Rustam Galimullin\textsuperscript{\rm 1}, Munyque Mittelmann\textsuperscript{\rm 2},  Laurent Perrussel\textsuperscript{\rm 3}
}
\affiliations{
    \textsuperscript{\rm 1}University of Bergen,  Norway    \\\textsuperscript{\rm 2}CNRS, LIPN, Université Sorbonne Paris Nord, France \\
    \textsuperscript{\rm 3}IRIT, Université Toulouse Capitole, France \\ 
    rustam.galimullin@uib.no, mittelmann@lipn.univ-paris13.fr, laurent.perrussel@irit.fr
}

\begin{document}

\maketitle

\begin{abstract}

In diffusion auctions, sellers can leverage an underlying social network to broaden participation, thereby increasing their potential revenue. 
Specifically, sellers can incentivise participants in their auction to diffuse information about the auction through the network. While numerous variants of such auctions have been recently studied in the literature, the formal verification and strategic reasoning perspectives have not been investigated yet.   

Our contribution is threefold. First, we introduce a logical formalism that captures the dynamics of diffusion and its strategic dimension. Second, for such a logic, we provide model-checking procedures that allow one to verify properties as the Nash equilibrium, and that pave the way towards checking the existence of sellers' strategies. Third, we establish computational complexity results for the presented algorithms. 
\end{abstract}



\section{Introduction}

In auction theory and mechanism design \cite{Nisan2007}, the set of participants is typically 
fixed and socially independent, in the sense that any underlying social network among agents is not taken into account. In contrast, by leveraging agents' social networks, a seller could use buyers' connections to promote the auction \cite{guo2021emerging}. This has a clear advantage: a larger market may include participants with higher valuations, leading to a potential increase in social welfare or the sellers' revenue. On the other hand, buyers act as competitors and have no incentives to invite more participants, as doing so would increase competition and reduce their likelihood of securing the item being auctioned. 

The challenge of encouraging participants to propagate the auction among their social connections has recently sparked interest in the mechanism design community \cite{zhao2021mechanism}.
In particular, it led to the introduction of \textit{diffusion auctions} \cite{zhao2018selling,LiHGZ22}, where sellers propose incentives to buyers so that they can 
benefit from inviting their neighbours. The intuition is that the mechanism guarantees that the buyer’s new utility after propagating the auction is not less than her utility of participating in the auction with the original participants. 
Their main benefit is the increase in the number of participants while guaranteeing economic properties such as incentive-compatibility or optimality \cite{ZhangZZ24}.  
Yet, two critical aspects remain unexplored --- the strategic behaviour of sellers in diffusion auctions, especially when multiple sellers 
compete to reach the most valuable buyers, and the formal verification of such mechanisms. 





In the last decades, a number of logics have been proposed to reason about agents' 
strategic capabilities with prime examples being 
\textit{Coalition Logic} (CL) \cite{pauly02},   \textit{Alternating-time Temporal Logic} (ATL) \cite{alur2002}, and \textit{Strategy Logic}  \cite{MogaveroMPV14}.
Combined with model-checking techniques \cite{clarke2018handbook}, these frameworks provide powerful tools for specification and verification of multi-agent systems, with applications to several problems, from   
the 
analysis of voting protocols \cite{BelardinelliCDJ21,JamrogaKM22} to the verification of auctions and  mechanism design \cite{MittelmannAIJ2025,MittelmannMMP23}.

In this paper, we provide a formal framework for the specification and verification of strategic properties in diffusion auctions. In doing so, we combine the intuitions from \textit{social network logics} \cite{minathesis}, \textit{dynamic epistemic logic} (DEL) \cite{hvdetal.del:2007}, as well as the aforementioned CL and ATL.  We believe that this is \textit{the first logic-based approach to formal verification of diffusion auctions and strategic abilities of sellers in them}.  


\paragraph{Contribution} 
We introduce the \textit{$n$-seller logic for diffusion incentives} $
\mathcal{L}^n$ and its \textit{strategic} variant $\mathcal{SL}^n$. These logics are interpreted on diffusion auction mechanism models that are quite general and thus capture a wide variety of mechanisms. Both $\mathcal{L}^n$ and $\mathcal{SL}^n$ allow us to capture the dynamics of diffusion of information about auctions and their strategic dimension. By `dynamics' here, we mean the change of the underlying social network as a result of sellers proposing incentives to buyers to invite their neighbours to an auction. For these logics, we provide model-checking procedures that allow one to verify properties such as Nash equilibrium, and that pave the way towards checking the existence of sellers' strategies. For the presented algorithms, we establish computational complexity results.


\section{Diffusion Auctions With Multiple Sellers}


We start by presenting a formal framework for multiple-seller auctions, where each seller is selling (a copy of) the same item. Sellers and buyers in such a setting are connected via an underlying social network, whose structure the sellers can try to exploit by incentivising their direct neighbours (i.e., buyers participating in the sellers' auctions) to invite all their friends to join the corresponding auction.

Let $\mathsf{S} = \{\sigma_1, \sigma_2, ..., \sigma_n\}$ be a finite non-empty set of $n$ names of sellers, and $\mathsf{B} = \{\beta_1, \beta_2, ...\}$ be a countable set of names of buyers, such that $\mathsf{S} \cap \mathsf{B} = \emptyset$.  Also, let $\mathsf{Nom} = \mathsf{S} \cup \mathsf{B}$ denote the total set of agent names, or nominals. We will also write $\mathsf{B}^\bullet = \mathsf{B} \cup \{\bullet\}$ and $\mathsf{Nom}^\bullet = \mathsf{Nom} \cup \{\bullet\}$, where nominal $\bullet$ intuitively stands for `the current agent'. Finally, let $\mathsf{Terms} = \{ut_\alpha \mid \alpha \in \mathsf{Nom}^\bullet\}$ 
be a set of terms denoting utilities of agents $ut_\alpha$ and a 
term $ut_\bullet$ denoting the utility of the current agent.

\begin{definition}
The language $\mathcal{L}^n$ of the \emph{$n$-seller logic for diffusion incentives} is defined by the following grammar:
    \begin{align*}
    &\varphi:= &&\alpha\mid 
    (z_1 t_1 + ...+z_m t_m) \geqslant z  \mid \lnot \varphi  \mid (\varphi \land \varphi) 
    \mid \square\varphi \mid \\ 
   & &&[\sigma_1: \beta_1, ..., \sigma_n: \beta_n] \varphi \mid \heartsuit\gamma,
\end{align*}
where $\alpha \in \mathsf{Nom}$, $z \in \mathbb{Z}$,
$t_i \in \mathsf{Terms}$,
$\sigma_i,\sigma_j \in \mathsf{S}$ with $\sigma_i \neq \sigma_j$, $\beta_i \in \mathsf{B}^\bullet$, and $\gamma \in \mathsf{Nom}^\bullet$. Here, $\square \varphi$ means `all friends of the current agent satisfy $\varphi$', and $\heartsuit \gamma$ means that agent named $\gamma$, which can be either a seller or a buyer, gets an item in the current configuration of a mechanism. 

Constructs $[\sigma_1: \beta_1, \sigma_2:\beta_2, ..., \sigma_n: \beta_n] \varphi$ (abbreviated as $[\overline{\sigma}:\overline{\beta}]\varphi$), where $n = |\mathsf{S}|$, capture the concurrent information diffusion about auctions. This is done by sellers $\sigma_i$ incentivising the respective buyers $\beta_i$, i.e., paying them some sum, to invite all their friends to join the seller's auction. 
Clause $\sigma_i: \bullet$ denotes the case in which seller $\sigma_i$ does not incentivise anyone, i.e., she does nothing or skips her turn.
Hence, we will write, e.g., $[\sigma_1: \beta_1, \sigma_2:\beta_2]\varphi$ if only sellers named $\sigma_1$ and $\sigma_2$ do not choose $\bullet$. Given $[\overline{\sigma}:\overline{\beta}]\varphi$, we will denote as $\overline{\sigma}^{\setminus \bullet}$ the set of agents $\sigma_i \in \overline{\sigma}$ such that the corresponding $\beta_i \neq \bullet$. Observe that even though this modality is concurrent, i.e., everyone is making moves in parallel, we can easily model consecutive moves by selecting $\bullet$ for all the agents who are not playing in the current turn. For the case of $1$-seller auctions $\mathcal{L}^1$, we will write $[\beta]\varphi$ instead of $[\sigma:\beta] \varphi$. Finally, having a sequence $\Upd = \sigma_1: \beta_1, ..., \sigma_n: \beta_n$, we will denote by $\Upd(\sigma_i)$ the corresponding buyer's name $\beta_i$.

Duals are defined as $\Diamond \varphi := \lnot \square \lnot \varphi$ and $\langle \overline{\sigma}:\overline{\beta} \rangle \varphi := \lnot [ \overline{\sigma}:\overline{\beta} ] \lnot \varphi$. For the linear inequalities\footnote{Originally, these linear inequalities in a logical context were used to capture reasoning about probabilities \cite{fagin90}. Recently, they were also used to express budgets and costs in dynamic epistemic logic \cite{DolgorukovGG24}. We follow the latter approach in this work. }, we can use the following abbreviations: $t_1 - t_2 \geqslant z$ for $t_1 + (-1)t_2 \geqslant z$,
$t_1 \geqslant t_2$ for $t_1 - t_2 \geqslant 0$,
$t_1\leqslant z$ for $-t_1\geqslant -z$, $t_1<z$ for $\neg (t_1 \geqslant z)$,
and $t_1 = z$ for $(t_1\geqslant z)\land (t_1\leqslant z)$. 
We can also use rational numbers in our formulas via abbreviations (e.g., $t\geqslant \frac{1}{2}$ is an
abbreviation for $2t\geqslant 1$).  
All other standard abbreviations of logic 
and the rules for removing parentheses hold.
\end{definition}

\begin{example}
With our language, we can express various desirable properties of mechanisms, both static and dynamic. The following formulas are examples for $\mathcal{L}^1$.
\begin{itemize}
    \item $ut_\alpha = 3 \land [\alpha ] (ut_\alpha > 3)$ for `the utility of agent $\alpha$ is 3, and after she was incentivised by the seller to invite her friends to participate in the auction, her utility increased'.
    \item $ut_\bullet = 5 \land \square (ut_\bullet \geqslant 5) \land \Diamond(\alpha \land \heartsuit \alpha)$ for `the utility of the current agent is 5, and all her friends have utilities of at least 5, and she also has a friend 
    $\alpha$ who gets an item'. 
\end{itemize}
\end{example}

Formulas of $\mathcal{L}^n$ are interpreted on diffusion auction mechanisms.

\begin{definition}
A \emph{market network with $n$ sellers} $\mathcal{M}$ is a tuple $(Agt, F, Bdg, V, I, N)$, where 
\begin{itemize}
    \item $Agt = B \cup S$ is the set of agents, where $B = \{a,b,c,...\}$ is a non-empty set of \emph{buyers}, and $S = \{s_1, ..., s_n\}$ is a non-empty set of \emph{sellers}, and $B \cap S = \emptyset$;
    \item $F:Agt \to 2^{B}$ is a symmetric irreflexive \emph{friendship} (neighbour) relation;
    \item $Bdg:Agt \to \mathbb{Q}^+ \cup \{0\} $ is a non-negative \emph{budget} for each agent;
    \item $V:B \to \mathbb{Q}^+ \cup \{0\}$ s.t. $V(a) \leqslant Bdg(a)$ assigns to each buyer a non-negative \emph{valuation} of the item being sold;
    \item $I:B\times S \to \mathbb{Q}^+ \cup \{0\}$ assigns to each buyer the non-negative \emph{incentive} that each seller is willing to pay to them to invite their friends; 
    \item $N  = N_{\mathsf{S}} \cup N_{\mathsf{B}}$ is a \emph{naming function}, where $N_{\mathsf{S}} : \mathsf{S} \to S$ and $N_{\mathsf{B}} : \mathsf{B} \to B$ are surjective functions\footnote{Observe that function $N$ is well-defined since $\mathsf{S} \cap \mathsf{B} = \emptyset$.}.
\end{itemize} 

An \emph{$n$-seller diffusion auction mechanism} ($n$-DAM, or DAM) $M$ is a tuple $(\mathcal{M}, P, \Pay, U)$, where 
\begin{itemize}
    \item $\mathcal{M}=(Agt, F, Bdg, V, I, N)$ is a market network with $n$ sellers;
    \item $P_{\mathcal{M}}: Agt \to \{0,1\}$   is the \emph{allocation (placement) function}, which specifies whether an agent receives an item in an auction conducted within the market network $\mathcal{M}$;
    \item $\Pay_{\mathcal{M}}: B \to \mathbb{Q}^+ \cup \{0\}$  is the \emph{payment function},  which specifies the value each buyer should pay in an auction run within the market network $\mathcal{M}$;
    \item $U_{\mathcal{M}}: Agt \to \mathbb{Q}^+ \cup \{0\}$ is the \emph{utility function}.
\end{itemize}

We omit subscripts $\mathcal{M}$ whenever it does not cause confusion. We write $M,a$ to refer to a specific agent $a$ in $M$.
\end{definition}

Observe that the definitions of the allocation, payment, and utility functions do not specify their exact details. This makes our definition of diffusion auction mechanisms general, allowing us to incorporate various types of mechanisms. The only restriction we put on the functions is that their complexity is no greater than the complexity of the model-checking problem of a given logic. As we will see later, model checking $\mathcal{L}^n$ is in P, and thus in this section we assume that $P$, $\Pay$, and $U$ are \emph{computable in polynomial time}. 
While the optimal allocation function in combinatorial auctions is NP-complete 
\cite{Nisan2007}\footnote{The optimal allocation function is used to compute allocations and payments
in the Vickrey–Clarke–Groves mechanism \cite{vickrey1961counterspeculation,clarke1971multipart,groves1973incentives}.  
}, several mechanisms whose functions are computable in polynomial time have been proposed. Those include a strategyproof combinatorial auction \cite{dobzinski2012computational}, a double auction mechanism in social networks \cite{xu2020design}, and  McAfee's double auction mechanism \cite{mcafee1992dominant}. 



For our running examples, we will use the \emph{single item, multiple units, first price} (SMF) auctions.
\begin{definition}[SMF Auction]
\label{def:SMF}
  Given a market network with $n$ sellers  $\mathcal{M} = (Agt, F, Bdg, V, I, N)$, the placement function $P$ is defined as follows.

  For a seller $s_i$, let $\mathfrak{s}_i$ be the ordered set of valuations $V(a)$ of buyers $a$ such that $a \in F(s_i)$\footnote{Observe that having sets $\mathfrak{s}_i$, and not multisets,  suffice as each buyer submits at most one valuation.}. The ordering is from the highest to the lowest valuations, where the ties are broken by the lexicographic ordering of the buyers. 
  This set is totally ordered. We denote the first element of $\mathfrak{s}_i$ as $\mathfrak{s}_i (1)$.

  We refine the sets $\mathfrak{s}_i$ to account for the distribution of items. A refinement of a totally ordered set $\mathfrak{s}_i$, denoted by $\overline{\mathfrak{s}}_i$, is defined as $\overline{\mathfrak{s}}_i = \mathfrak{s}_i \setminus \{\mathfrak{s}_j (1)| 0 < j < i\}$. Intuitively, a refinement is the set of bidders in the current auction minus those bidders that are already getting the item from some other auction.  
Finally, for all $s_i \in S$, if $\overline{\mathfrak{s}}_i$ is non-empty, then $P(a) = 1$ and $P(s_i) = 0$ for $V(a) = \overline{\mathfrak{s}}_i (1)$, and $P(s_i) = 1$ otherwise. Observe that in SMF auctions, it is implied that buyers strive to acquire only one (copy of the) item, and hence some sellers may end up not selling their items.

  As for the $\Pay$ function, buyers, if they get an item, pay the amount equal to their valuation of the item.
Function $U$ is defined as follows: for a buyer $a \in B$, $U(a) = Bdg(a) - V(a)\cdot P(a)$; 
for the seller $s_i$, $U(s_i) = V(a) + Bdg(s_i)$, where $V(a) = \overline{\mathfrak{s}}_i (1)$, and $U(s_i) = Bdg(s_i)$ if $\overline{\mathfrak{s}}_i$ is empty.
\end{definition}

\begin{example}
As an example of how the placement function works in Definition \ref{def:SMF}, consider two mechanisms in Figure \ref{fig:placement}. In both mechanisms, all valuations and incentive values are identical. Moreover, $\heartsuit$ denotes the allocation of items.

\begin{figure}[h!]
\centering
\scalebox{0.8}{
   \begin{tikzpicture}
\node[circle,draw=black, minimum size=4pt,inner sep=0pt, fill = black, label=above:{$s_1$}](s1) at (0,0) {};
\node[circle,draw=black, minimum size=4pt,inner sep=0pt,, fill = black , label=left:{$b_1:\heartsuit$}](d) at (0,-2) {};
\node[circle,draw=black, minimum size=4pt,inner sep=0pt, , fill = black, label=above:{$s_2$}](a) at (2,0) {};
\node[circle,draw=black, minimum size=4pt,inner sep=0pt, , fill = black, label=right:{$b_2:\heartsuit$}](e) at (2,-2) {};

\draw[thick] (s1) to (d);
\draw[thick] (s1) to (e);
\draw[thick] (a) to (e);
\draw[thick] (a) to (d);

\end{tikzpicture}
  \begin{tikzpicture}
\node[circle,draw=black, minimum size=4pt,inner sep=0pt, fill = black, label=above:{$s_1$}](s1) at (0,0) {};
\node[circle,draw=black, minimum size=4pt,inner sep=0pt,, fill = black , label=right:{$b:\heartsuit$}](d) at (1,-2) {};
\node[circle,draw=black, minimum size=4pt,inner sep=0pt, , fill = black, label=above:{$s_2:\heartsuit$}](a) at (2,0) {};

\draw[thick] (s1) to (d);
\draw[thick] (d) to (a);

\end{tikzpicture}
}
\caption{Mechanisms $M_1$ (left) and $M_2$ (right).}
   \label{fig:placement}
\end{figure}
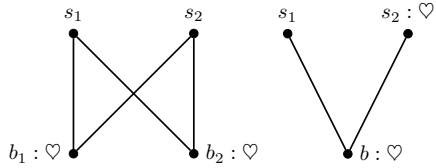 

In $M_1$ we have that both buyers participate in both auctions run by $s_1$ and $s_2$, and each buyer wins an item. For the 
sellers, sets $\mathfrak{s}_1$ and $\mathfrak{s}_2$ are $\{V(b_1), V(b_2)\}$. Recall that $V(b_1) = V(b_2)$ and we use the lexicographic preference ordering to break ties. The refinement $\overline{\mathfrak{s}}_1 = \{V(b_1), V(b_2)\}$, and the refinement $\overline{\mathfrak{s}}_2 = \{V(b_1), V(b_2)\} \setminus \{V(b_1)\} = \{V(b_2)\}$. Hence, by the definition of the placement function, $P(b_1) = 1$, $P(b_2) = 1$, $P(s_1) = 0$, and $P(s_2) = 0$.

Now, let us take a look at mechanism $M_2$, where we again have two sellers but now only one buyer. There, $\mathfrak{s}_1 = \mathfrak{s}_2 = \{V(b)\}$, and the corresponding refinements are $\overline{\mathfrak{s}}_1 = \{V(b)\}$ and $\overline{\mathfrak{s}}_2 = \{V(b)\} \setminus \{V(b)\} = \emptyset$ (using lexigraphic ordering over sellers as a tie-breaking rule). Hence, the placement function is $P(b) = 1$,  $P(s_1) = 0$, and $P(s_2) = 1$. Note that following Definition \ref{def:SMF},  seller $s_2$ keeps her item. 
\end{example}

Since sellers can incentivise buyers to invite their friends to join an auction, our mechanisms are \textit{dynamic}, i.e., the social network structure changes as a result of sellers' actions. To capture this, we define the \textit{update} of a mechanism.  In our definition, we assume that each buyer can be incentivised by one seller at a time. And, moreover, since buyers are rational, if they are offered incentives from two separate sellers, they choose the higher incentive, and their friends join the auction of the seller offering the higher incentive. 

\begin{definition}
    Let an $n$-DAM $M = ((Agt, F, Bdg, V, I, N),$ $P,$ $\Pay,$ $U)$ and a sequence $\Upd = \sigma_1: \beta_1, ..., \sigma_n: \beta_n$ be given. An $n$-DAM \emph{updated by the concurrent invitations by $n$ sellers} 
    $M^{\Upd}$ is a tuple $((Agt, F^{\Upd}, Bdg^{\Upd}, V, I, N), P, \Pay, U)$, where for all $s \in S$, if $N(\sigma_i) = s$, $\Upd(\sigma_i) = \beta_i \neq \bullet$,   and $s = \arg\,\max_{s\in S}I(N(\beta_i),s')$, then
    \begin{itemize}
        \item  
        $F^{\Upd}(s) = F(s) \cup \{b | b \in B \text{ and }b \in F(N(\beta_i))\}$,
        \item 
        $Bdg^{\Upd}(s) =$ $Bdg(s) - I(N(\beta_i),s)$ and $Bdg^{\Upd}(N(\beta_i)) =$ $Bdg(N(\beta_i)) + I(N(\beta_i),s)$.
    \end{itemize}
    Ties are broken by the lexicographic order of sellers' names\footnote{In particular, if there are two or more sellers offering a buyer the same maximal incentive, the buyer propagates the auction information of the seller that appears first in the lexicographic order.}.
\end{definition}

Intuitively, given a concurrent information diffusion operator $\overline{\sigma}: \overline{\beta}$, in the updated mechanism all friends of buyer $\beta_i$ will join the auction run by $\sigma_i$ if $\sigma_i$ is one of the sellers that incentivise $\beta_i$, and moreover, offers the highest incentive. Then, the seller's budget is reduced by the value of the incentive, and the budget of the corresponding buyer is increased by the value of the incentive.    

\begin{definition}
Let $M = ((Agt, F, Bdg, V, I, N),$ $P, \Pay, U)$ be an $n$-DAM. The semantics of $\mathcal{L}^n$  is defined by induction as follows: 
    \begin{alignat*}{1}
        &M, a \models \alpha  \text{ iff }  N(\alpha) = a\\
        &M, a \models \lnot \varphi  \text{ iff }  M, a \not \models \varphi\\
        &M, a \models\varphi \land \psi  \text{ iff }  M, a \models \varphi \text{ and } M, a \models \psi\\
        &M, a \models \square \varphi  \text{ iff }  \forall b \in Agt: b \in F(a) \text{ implies } M, b \models \varphi\\
&M, a \models [\overline{\sigma}:\overline{\beta}] \varphi  \text{ iff }  \text{if }  \forall \sigma_i \in \overline{\sigma}^{\setminus \bullet}: 
        N(\beta_i) \in F(N(\sigma_i)) \text{ and }\\ 
        &\quad \quad Bdg(N(\sigma_i)) \geqslant I(N(\beta_i),N(\sigma_i)), \text{ then } M^{\overline{\sigma}:\overline{\beta}},a \models \varphi\\
        &M,a \models \heartsuit \alpha  \text{ iff }  
        \begin{cases}
        P(N(\alpha)) = 1 &\text{ if } \alpha \neq \bullet,\\
        P(a) = 1 &\text{ if } \alpha = \bullet\\
        \end{cases}\\
        &M,a \models \sum_{i = 1}^m z_i t_i  \geqslant z  \text{ iff }  \sum_{i = 1}^m z_i t'_i \geqslant z, \text{ where } \\
        &\qquad \quad t'_i =
    \begin{cases}
        U(N(\alpha)) &\text{ if } t_i = ut_\alpha,\\
         U(a) &\text{ if } t_i = ut_{\bullet}.
    \end{cases}
\end{alignat*}     

\end{definition}

The semantics definition for the clause $M, a \models [\overline{\sigma}:\overline{\beta}] \varphi$ checks whether all sellers, who did not choose to skip their turn, are (i) incentivising the buyers that currently participate in their auction, and (ii) whether the sellers have sufficient budgets to pay the incentives. If (i) and (ii) hold, then $\varphi$ is evaluated in the updated mechanism  $M^{\overline{\sigma}:\overline{\beta}},a$. 
The dual $\langle \overline{\sigma}:\overline{\beta} \rangle \varphi$ holds iff (i) and (ii) hold, and $M^{\overline{\sigma}:\overline{\beta}},a \models \varphi$.
For clauses of $\heartsuit \alpha$ and linear inequalities, we distinguish the cases of an agent named $\alpha$ and the current agent, denoted by nominal $\bullet$. 


\begin{example}
    Consider  1-seller SMF mechanism $M$ and its update $M^{\sigma:\alpha}$ in Figure \ref{fig:example1}.
    \begin{figure}[t!]
\centering
\scalebox{0.8}{
\begin{tikzpicture}
\node[circle,draw=black, minimum size=4pt,inner sep=0pt, fill = black, label=below:{$s:\sigma, 5$}](1) at (4.5,0) {};
\node[circle,draw=black, minimum size=4pt,inner sep=0pt,, fill = black , label=above:{$a:\alpha, (3, 1,5)$}](2) at (3,1) {};

\node[circle,draw=black, minimum size=4pt,inner sep=0pt, , fill = black, label=below:{$b:\beta, (8,2,6)$}](3) at (3,-1) {};
\node[circle,draw=black, minimum size=4pt,inner sep=0pt, , fill = black, label=above:{$c:\gamma, (9, 9,1)$}](4) at (6,1) {};
\node[circle,draw=black, minimum size=4pt,inner sep=0pt, , fill = black, label=below:{$d:\delta, (11,10,0)$}](5) at (6,-1) {};

\draw [thick](1) to (2);
\draw [thick](1) to (3);
\draw [thick](2) to (3);
\draw [thick](2) to (4);
\draw [thick](5) to (4);

\end{tikzpicture}
\begin{tikzpicture}
\node[circle,draw=black, minimum size=4pt,inner sep=0pt, fill = black, label=below:{$s:\sigma, 0$}](1) at (4.5,0) {};
\node[circle,draw=black, minimum size=4pt,inner sep=0pt,, fill = black , label=above:{$a:\alpha, (8, 1,5)$}](2) at (3,1) {};

\node[circle,draw=black, minimum size=4pt,inner sep=0pt, , fill = black, label=below:{$b:\beta, (8, 2,6)$}](3) at (3,-1) {};
\node[circle,draw=black, minimum size=4pt,inner sep=0pt, , fill = black, label=above:{$c:\gamma, (9, 9,1)$}](4) at (6,1) {};
\node[circle,draw=black, minimum size=4pt,inner sep=0pt, , fill = black, label=below:{$d:\delta, (11, 10,0)$}](5) at (6,-1) {};

\draw [thick](1) to (2);
\draw [thick](1) to (3);
\draw [thick](2) to (3);
\draw [thick](2) to (4);
\draw [thick](5) to (4);
\draw [thick,dashed](1) to (4);

\end{tikzpicture}
}

\caption{Mechanism $M$ (left) and updated mechanism $M^{\sigma:\alpha}$ (right). For the seller $s$, her name is $\sigma$ and her budget in $M$ is 5. For buyer $a$, her name is $\alpha$, and $(3, 1,5)$ denotes the fact that $Bdg(a) = 3$, $V(a) = 1$, and $I(a, s) = 5$. Similarly, for other agents. The new link in $M^\alpha$ is dashed.}
\label{fig:example1}
\end{figure}
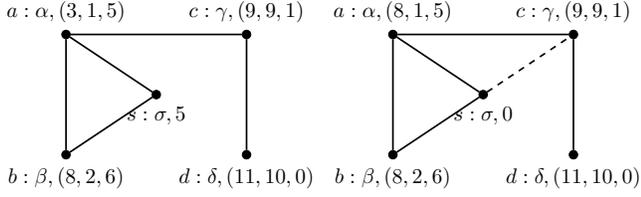 
    We have, for example, that $M,a \models (ut_\sigma = 7) \land \heartsuit \beta \land \langle \alpha \rangle (ut_\sigma = 9 \land \heartsuit \gamma)$, where the utility of the seller $\sigma$ increases after she incentivises the agent named $\alpha$ to invite her friends. Indeed, in mechanism $M$, the seller's budget is 5, and the highest valuation among the buyers participating in the auction is 2 (from agent $\beta$). Hence, the total utility of the seller is 7. Since agent $\beta$ is currently the highest bidder, she would have been the winner of the auction in the current mechanism (conjunct $\heartsuit \beta$). After the seller incentivises $\alpha$ to invite all her friends to 
    the auction, the budget of $\alpha$ increases by her incentive (i.e. by $5$ to the total 8), and agent $\gamma$ joins the auction (dashed line in $M^{\sigma:\alpha}$). Since her valuation 
    is the highest, she gets the item (conjunct $\heartsuit \gamma$) and the utility of the seller increases to 9. 
    We also have that $M,a \models \lnot \langle \alpha \rangle \langle \gamma \rangle (ut_\sigma >9 )$, i.e., the seller cannot increase her utility further by trying to reach the richest buyer $\delta$ as she does not have enough budget for two rounds of referrals.

    Now we turn to a 2-seller example. Consider the SMF mechanism $M$ and its updates in Figure \ref{fig:example2}. 
We assume that sellers $s_1$ and $s_2$ have both budgets 1, and that for both of them and all buyers, the incentives are 1, i.e., each seller can incentivise only one 
buyer. Moreover, let us assume that each buyer, apart from $b$ and $e$, evaluates the item as 1. Buyer $b$ has valuation 4, and buyer $e$ has valuation 3.

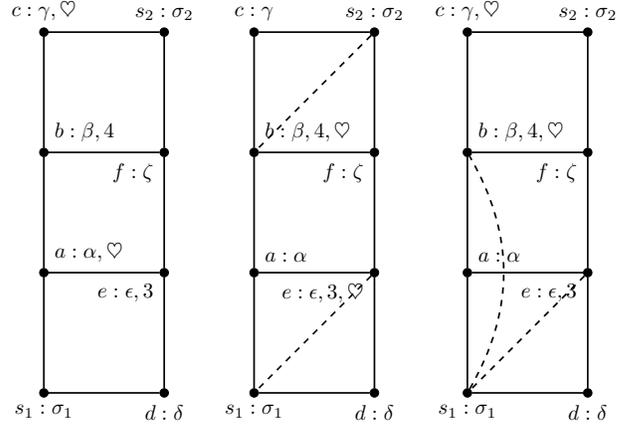
\begin{figure}
\centering
\scalebox{0.8}{
   \begin{tikzpicture}[rotate = 90]
\node[circle,draw=black, minimum size=4pt,inner sep=0pt, fill = black, label=below:{$s_1:\sigma_1$}](s1) at (0,0) {};
\node[circle,draw=black, minimum size=4pt,inner sep=0pt,, fill = black , label=below:{$d:\delta$}](d) at (0,-2) {};
\node[circle,draw=black, minimum size=4pt,inner sep=0pt, , fill = black, label=above right:{$a:\alpha, \heartsuit$}](a) at (2,0) {};
\node[circle,draw=black, minimum size=4pt,inner sep=0pt, , fill = black, label=below left:{$f:\zeta$}](f) at (4,-2) {};
\node[circle,draw=black, minimum size=4pt,inner sep=0pt, , fill = black, label=below left:{$e:\epsilon, 3$}](e) at (2,-2) {};
\node[circle,draw=black, minimum size=4pt,inner sep=0pt, fill = black, label=above:{$s_2:\sigma_2$}](s2) at (6,-2) {};
\node[circle,draw=black, minimum size=4pt,inner sep=0pt, , fill = black, label=above:{$c:\gamma,\heartsuit$}](c) at (6,0) {};
\node[circle,draw=black, minimum size=4pt,inner sep=0pt, , fill = black, label=above right:{$b:\beta,4$}](b) at (4,0) {};

\draw[thick] (s1) to (d);
\draw[thick] (s1) to (a);
\draw[thick] (d) to (e);
\draw[thick] (e) to (a);

\draw[thick] (s2) to (f);
\draw[thick] (s2) to (c);
\draw[thick] (b) to (c);
\draw[thick] (f) to (e);

\draw[thick] (a) to (b);
\draw[thick] (f) to (b);
\end{tikzpicture}
}
\scalebox{0.8}{
\begin{tikzpicture}[rotate = 90]
\node[circle,draw=black, minimum size=4pt,inner sep=0pt, fill = black, label=below:{$s_1:\sigma_1$}](s1) at (0,0) {};
\node[circle,draw=black, minimum size=4pt,inner sep=0pt,, fill = black , label=below:{$d:\delta$}](d) at (0,-2) {};
\node[circle,draw=black, minimum size=4pt,inner sep=0pt, , fill = black, label=above right:{$a:\alpha$}](a) at (2,0) {};
\node[circle,draw=black, minimum size=4pt,inner sep=0pt, , fill = black, label=below left:{$f:\zeta$}](f) at (4,-2) {};
\node[circle,draw=black, minimum size=4pt,inner sep=0pt, , fill = black, label=below left:{$e:\epsilon, 3, \heartsuit$}](e) at (2,-2) {};
\node[circle,draw=black, minimum size=4pt,inner sep=0pt, fill = black, label=above:{$s_2:\sigma_2$}](s2) at (6,-2) {};
\node[circle,draw=black, minimum size=4pt,inner sep=0pt, , fill = black, label=above:{$c:\gamma$}](c) at (6,0) {};
\node[circle,draw=black, minimum size=4pt,inner sep=0pt, , fill = black, label=above right:{$b:\beta,4, \heartsuit$}](b) at (4,0) {};

\draw[thick] (s1) to (d);
\draw[thick] (s1) to (a);
\draw[thick] (d) to (e);
\draw[thick] (e) to (a);

\draw[thick] (s2) to (f);
\draw[thick] (s2) to (c);
\draw[thick] (b) to (c);
\draw[thick] (f) to (e);

\draw[thick] (a) to (b);
\draw[thick] (f) to (b);
\draw[thick,dashed] (s2) to (b);
\draw[thick,dashed] (e) to (s1);

\end{tikzpicture}
}
\scalebox{0.8}{
  \begin{tikzpicture}[rotate = 90]
\node[circle,draw=black, minimum size=4pt,inner sep=0pt, fill = black, label=below:{$s_1:\sigma_1$}](s1) at (0,0) {};
\node[circle,draw=black, minimum size=4pt,inner sep=0pt,, fill = black , label=below:{$d:\delta$}](d) at (0,-2) {};
\node[circle,draw=black, minimum size=4pt,inner sep=0pt, , fill = black, label=above right:{$a:\alpha$}](a) at (2,0) {};
\node[circle,draw=black, minimum size=4pt,inner sep=0pt, , fill = black, label=below left:{$f:\zeta$}](f) at (4,-2) {};
\node[circle,draw=black, minimum size=4pt,inner sep=0pt, , fill = black, label=below left:{$e:\epsilon, 3$}](e) at (2,-2) {};
\node[circle,draw=black, minimum size=4pt,inner sep=0pt, fill = black, label=above:{$s_2:\sigma_2$}](s2) at (6,-2) {};
\node[circle,draw=black, minimum size=4pt,inner sep=0pt, , fill = black, label=above:{$c:\gamma, \heartsuit$}](c) at (6,0) {};
\node[circle,draw=black, minimum size=4pt,inner sep=0pt, , fill = black, label=above right:{$b:\beta,4, \heartsuit$}](b) at (4,0) {};

\draw[thick] (s1) to (d);
\draw[thick] (s1) to (a);
\draw[thick] (d) to (e);
\draw[thick] (e) to (a);

\draw[thick] (s2) to (f);
\draw[thick] (s2) to (c);
\draw[thick] (b) to (c);
\draw[thick] (f) to (e);

\draw[thick] (a) to (b);
\draw[thick] (f) to (b);
\draw[thick,dashed, bend right] (s1) to (b);
\draw[thick,dashed] (e) to (s1);

\end{tikzpicture}
   }
\caption{Mechanism $M$ and its updates. In the mechanisms, there are two sellers, $s_1$ and $s_2$, and six buyers, $a,b,...,f$. Agents' names are shown near the state label. New links are dashed. The allocation of the sold items is depicted by $\heartsuit$.}
\label{fig:example2}
\end{figure} 

We can see that in mechanism $M$, each seller has utility of 2, and buyers $a$ and $c$ get the item (recall that we assume the lexicographic tie-breaking rule). Formally, $ut_{\sigma_1} = 2 \land ut_{\sigma_2} = 2 \land \heartsuit \alpha \land \heartsuit \gamma$. 
We can also consider a more cooperative goal of the mechanism configuration, where each agent either has the item or has a friend who has the item, expressed by $\bigwedge_{i \in \mathsf{Nom}_M} i \to (\heartsuit \bullet \lor \, \Diamond \heartsuit \bullet)$, where $\mathsf{Nom}_M = \{\sigma_1, \sigma_2, \alpha, ..., \zeta\}$. Mechanism $M$ does not satisfy this goal, as, for example, the formula does not hold for agent $d$.

Now, assume that both sellers decide to diffuse the information about their auctions over the network. Seller $s_1$ incentivises buyer $d$ to invite all her friends, i.e., $e$, to the auction; and similarly for seller $s_2$ and buyer $c$ who invites $b$. The resulting update $M^{\sigma_1:\delta, \sigma_2: \gamma}$ is depicted in the middle of Figure \ref{fig:example2}. After this diffusion, the utility of seller $s_1$ becomes 3 (1 is spent on the incentive, and agent $e$ bids 3), and the utility of seller $s_2$ is 4, i.e., for both sellers, the diffusion allowed them to increase their utilities by reaching bidders with higher valuations. Formally, we can write this as $ut_{\sigma_1} = 2 \land ut_{\sigma_2} = 2 \land [\sigma_1:\delta, \sigma_2:\gamma] (ut_{\sigma_1} > 2 \land ut_{\sigma_2} > 2)$.

Interestingly, such a diffusion performed by both sellers allows us to reach the cooperative goal: formula $[\sigma_1:\delta, \sigma_2:\gamma]\bigwedge_{i \in \mathsf{Nom}_M} i \to (\heartsuit \bullet \lor \,\Diamond \heartsuit \bullet)$ is valid in mechanism $M$.

Although the diffusion operator $[\sigma_1:\delta, \sigma_2:\gamma]$ allows sellers to increase their utilities and also achieve the cooperative goal, it requires a coordinated action from the sellers. Thus, seller $s_1$ can increase her utility even more and outperforms seller $s_2$ by incentivising buyer $a$ instead of $d$, and the resulting updated mechanism $M^{\sigma_1:\alpha, \sigma_2: \gamma}$ is depicted on the right of Figure \ref{fig:example2}. In this case, $a$ will invite her friends, including $b$, to join the auction of seller $s_1$. Observe that 
there is nothing seller $s_2$ can do to make her utility greater than 1 and thus outperform $s_1$ (recall that as a tie-breaking rule we use the lexicographic ordering, and hence, if buyer $b$ is being invited to both auctions, she will choose seller $s_1$). 
In formulas, $\bigwedge_{i \in \mathsf{Nom}_M \cup \{\bullet\}}[\sigma_1 : \alpha, \sigma_2: i] (ut_{\sigma_1} > ut_{\sigma_2})$. The cooperative goal $\bigwedge_{i \in \mathsf{Nom}_M} i \to (\heartsuit \bullet \lor \, \Diamond \heartsuit \bullet)$ is no longer valid in the mechanism as agent $d$ does not satisfy it.


\end{example}

\section{Strategic Properties of Diffusion Auctions}
Let us now demonstrate how to express Nash equilibrium in our setting, as well as study the complexities of the model checking and strategy existence problem. 

We first define an appropriate notion of a \textit{finite} DAM. The nuance here is that even if the set of agents is finite, we have a countably infinite number of names that can be assigned to agents. Luckily for us, whenever a problem requires both a mechanism and a formula in the input, we need to care about only those nominals that occur explicitly in the formula. 

Let $name (\varphi) $ be the finite set of nominals appearing in $\varphi$. 
Given a DAM $M = ((Agt, F, Bdg, V, I, N), P, \Pay, U)$, we define 
function $\Nfin: Agt \to \mathsf{Nom}$ that will pick for each agent a \textit{finite number} of names. Since $N$ is surjective, for each buyer $b$, the preimage $N^{-1} (b)$ is non-empty (and potentially infinite). Hence, for each buyer $b$, $\Nfin(b)$ picks all $\beta$ s.t. $\beta \in name(\varphi)$ and $N(\beta) = b$, and any single element of $N^{-1} (b)$ otherwise. For such a function, we have that $N(\Nfin (b)) = b$, and by construction, $\Nfin$ is finite. We will denote the mechanism $M$ with $N$ substituted for $\Nfin$ as $\Mfin$.
Intuitively, $\Nfin$ is a finite inverse of $N$ that picks for each agent a single name out of possibly an infinite number of names of the agent while also respecting the nominals appearing 
in $\varphi$.

It is straightforward to show the following by induction on $\varphi$ and using the definition of the semantics.
\begin{proposition}
Let $\varphi \in \mathcal{L}^n$ and a mechanism $M$ be given. Then we have that $M,a \models \varphi$ iff $\Mfin,a \models \varphi$.
\end{proposition}

The \emph{size} of mechanism $M$ 
is $|M| = |Agt| + |F| + |Bdg| + |V| + |I| + |\Nfin| + |P| + |\Pay| + |U|$.
Mechanism $M$ \emph{finite}, if $|M|$ is finite. Since all mechanisms in this section are finite, we will 
use interchangeably $M$ and $\Mfin$, and $N$ and $\Nfin$.

\paragraph{Nash Equilibrium} 
We can now verify that a given joint diffusion action is a one-step Nash equilibrium (NE) over a given finite mechanism $M$. This is significant, as it allows reasoning about optimal strategies of sellers. 
Let $\varphi := \langle \overline{\sigma}: \overline{\beta}\rangle \bigwedge_{1\leqslant i \leqslant n} ut_{\sigma_i} = m_i$ express that after a joint diffusion action $\langle \overline{\sigma}: \overline{\beta}\rangle$, utilities of all sellers $i$ are $m_i$. 

In order to verify that such a diffusion is indeed an NE, we check that no single seller can increase her utility by deviating, i.e., by incentivising some other buyer present in the same mechanism. 
The formula for the NE in this case is 
$\varphi \land \bigwedge_{i = 1}^n \bigwedge_{\gamma \in \mathsf{Nom}^{\bullet}_M} \langle\sigma_1: \beta_1, ..., \sigma_i:\gamma, ..., \sigma_n: \beta_n \rangle (ut_{\sigma_i} \leqslant m_i)$,
where $\mathsf{Nom}^{\bullet}_M$ is the set of names appearing in the finite mechanism, i.e., the range of $\Nfin$ plus $\bullet$.

We can further generalise the setting to a $k$-step NE, by taking $\varphi^k := \langle \overline{\sigma}: \overline{\alpha}\rangle^1 ... \langle \overline{\sigma}: \overline{\gamma}\rangle^k \bigwedge_{1\leqslant i \leqslant n} ut_{\sigma_i} = m_i$ meaning that after a $k$-sequence of joint diffusion actions, sellers' utilities are $m_i$'s. Then the corresponding formula for the $k$-step NE is $\varphi^k \land \bigwedge_{i = 1}^n \bigwedge_{\gamma \in \mathsf{Nom}^{\bullet}_M} \langle\sigma_1: \alpha_1, ..., \sigma_i:\gamma, ..., \alpha_n: \beta_n \rangle^1 .... \bigwedge_{i = 1}^n \bigwedge_{\gamma \in \mathsf{Nom}^{\bullet}_M} \langle\sigma_1: \beta_1, ..., \sigma_i:\gamma, ..., \sigma_n: \beta_n \rangle^k (ut_{\sigma_i} \leqslant m_i)$.

\subsection{Model Checking and Strategy Existence}
First, we show that the complexity of the model checking problem for $\mathcal{L}^n$ is in P.

\begin{theorem}
Model checking $\mathcal{L}^n$ is in P for the class of finite DAMs with polynomially computable placement, payment, and utility functions. 
\end{theorem}

\begin{proof}
In Algorithm \ref{diffMC}, we focus on the dynamic modality and the allocation operator\footnote{Modal, nominal \cite{franceschet06}, and arithmetic cases \cite{algorithms} are standard and can be computed in polynomial time. We omit them for brevity.}. 
\begin{breakablealgorithm}
	\caption{An algorithm for model checking $\mathcal{L}^n$}\label{diffMC} 
	\small
	\begin{algorithmic}[1] 		
		\Procedure{MC}{$M, a, \varphi$}		
\Case {$\varphi = [\overline{\sigma}:\overline{\beta}] \psi$} 
\For{$\sigma_i \in \overline{\sigma}^{\setminus \bullet}$}
\If {$N(\beta_i) \in F(N(\sigma_i))
$ and $ Bdg(N(\sigma_i)) \geqslant   $  
\Statex  \hspace{\algorithmicindent}\hspace{\algorithmicindent}\quad\;\;
$I(N(\beta_i), N(\sigma_i))$
 }
\If{$N({\sigma_i}) = \arg\,\max_{N({\sigma_i}) \in S} I (N({\beta_i}),N({\sigma_i})))$ 
\Statex  \hspace{\algorithmicindent}\hspace{\algorithmicindent}\qquad\;\; 
}
\State{
$\begin{aligned}
F^{\overline{\sigma}:\overline{\beta}} (N(\sigma_i)) = F(N(\sigma_i)) \cup \{a \mid 
a \in B \\ 
\text{ and } a \in F(N(\beta_i))\}
\end{aligned}$
} 
\State{
$\begin{aligned}
Bdg^{\overline{\sigma}:\overline{\beta}} (N(\sigma_i)) = \; & Bdg(N(\sigma_i)) \;- 
\\
& I (N(\beta_i), N(\sigma_i))
\end{aligned}$
}
\State{$
\begin{aligned}
    Bdg^{\overline{\sigma}:\overline{\beta}} (N(\beta_i))) = \; & Bdg(N(\beta_i)))\; + 
    \\
    & I (N(\beta_i),N(\sigma_i)))
\end{aligned} 
$}
\EndIf
\Else
\State{\textbf{return} $true$}
\EndIf
\EndFor
\State{\textbf{return} $\textsc{MC} (M^{\overline{\sigma}:\overline{\beta}}, a, \psi)$}
\EndCase

\Case{$\varphi = \heartsuit \gamma$}
\If{$\gamma \neq \bullet$}
\State{\textbf{return} $P(N(\gamma)) = 1$}
\Else
\State{\textbf{return} $P(a) = 1$}
\EndIf
\EndCase
   \EndProcedure

	\end{algorithmic}
\end{breakablealgorithm}

On line 5, we check that seller $N(\sigma_i)$ offers the highest incentive to the corresponding buyer. If there are two sellers that offer the highest incentive to the same corresponding buyer, we assume that $\arg\,\max_{N(\sigma_i) \in S} I (N(\beta_i),N(\sigma_i)))$ returns the seller that appears earlier in the lexicographic order.

The algorithm directly mimics 
the semantics of $\mathcal{L}^n$ and thus correctness can be shown by induction on $\varphi$. 
First, recall that here we assume that the allocation, payment, and utility functions are all computable in polynomial time. For the case of $\varphi = [\overline{\sigma}:\overline{\beta}] \psi$, the size of an updated mechanism $M^{\overline{\sigma}:\overline{\beta}}$ is at most $\mathcal{O}(|M|^2)$ (in the worst case, the friendship relation is universal). The procedure \textsc{MC}$(M, a, \varphi)$ is run for at most $|\varphi|$ times and for at most $|\varphi|$ mechanisms. Hence, \textsc{MC}$(M, a, \varphi)$ is used for a polynomial amount of time.
\end{proof}

Having the model checking result at hand, we can formulate and show the complexity of the strategy existence problem (proof is given in the Technical Appendix). The problem intuitively consists in checking whether, for a given mechanism and a mutual sellers' goal, 
there is a way (a strategy) for all sellers to achieve $\varphi$ in a finite number of steps. 

\begin{definition}
    Given a finite mechanism $M$ and a goal $\varphi \in \mathcal{L}^n$, the \emph{strategy existence problem} consists in determining whether there is a finite sequence of concurrent incentivisations $\langle \overline{\sigma}: \overline{\beta} \rangle^\ast = \langle \overline{\sigma}: \overline{\alpha} \rangle .... \langle \overline{\sigma}: \overline{\gamma} \rangle$ such that $M,s \models \langle \overline{\sigma}: \overline{\beta} \rangle^\ast \varphi$ for sellers $s \in S$.   
\end{definition}
\begin{restatable}{theorem}{strex}
\label{thm:strex}
    The strategy existence problem is NP-complete for the class of finite DAMs with polynomially computable placement, payment, and utility functions. 
\end{restatable}

\section{Reasoning About Sellers' Strategies}
While considering the strategy existence problem, we looked at how \textit{all} sellers can reach their joint goal via a sequence of concurrent incentivisation actions. However, in diffusion auctions with multiple sellers selling the copy of the same item, sellers and \textit{coalitions} thereof may compete against each other for buyers. To capture this strategic competitive setting,  
we introduce a modality inspired by coalitional operators from 
CL \cite{pauly02} and 
ATL \cite{alur2002}. In particular, we extend the language of $\mathcal{L}^n$ with $\langle \! [ \mathsf{C} ] \! \rangle \varphi$ for $\mathsf{C} \subseteq \mathsf{S}$, meaning that there is a (one-step) strategy for the coalition of sellers $\mathsf{C}$ to incentivise buyers such that no matter what other sellers do, $\varphi$ holds. In other words, modalities $\langle \! [ \mathsf{C} ] \! \rangle$ capture the ability of sellers in $\mathsf{C}$ to reach 
outcome $\varphi$ in the competitive setting.

\begin{definition}
The language $\mathcal{SL}^n$ of the \emph{$n$-seller strategic logic for diffusion incentives} 
is defined as follows:
    \begin{align*}
    &\varphi:= &&\alpha\mid 
    (z_1 t_1 + ...+z_m t_m) \geqslant z  \mid \lnot \varphi  \mid (\varphi \land \varphi) 
    \mid \square\varphi \mid \\
    & &&[\sigma_1: \beta_1, ..., \sigma_n: \beta_n] \varphi \mid [ \! \langle \mathsf{C} \rangle \! ] \varphi \mid  \heartsuit\gamma,
\end{align*}
where $\mathsf{C} \subseteq \mathsf{S}$, and the dual of $[ \! \langle \mathsf{C} \rangle \! ] \varphi$ is $\langle \! [\mathsf{C}] \! \rangle \varphi:= \lnot [ \! \langle \mathsf{C} \rangle \! ] \lnot \varphi$. We will denote by $\overline{\sigma_{\mathsf{C}}}: \overline{\beta_{\mathsf{C}}}$ the assignment of buyers names $\overline{\beta_{\mathsf{C}}}$ from $\mathsf{B}^\bullet$ to the coalition of sellers $\overline{\sigma_{\mathsf{C}}}$ such that all sellers not in coalition are assigned $\bullet$ (i.e., they skip their turn). Moreover, we will denote ${\overline{\sigma_\mathsf{C}}\cup \overline{\sigma_{\mathsf{S \setminus C}}}:\overline{\beta_\mathsf{C}}\cup \overline{\beta_{\mathsf{S \setminus C}}}}$ the full assignment of $n$ buyers from $\mathsf{B}^\bullet$ to $n$ sellers.
\end{definition}

\begin{definition}
Let $M = ((Agt, F, Bdg,$ $V, I, N),$ $P,$ $\Pay,$ $U)$ be an $n$-DAM. The semantics of $\mathcal{SL}^n$ extends the semantics of $\mathcal{L}^n$ with the following clause and its dual:
    \begin{alignat*}{3}
        &M,a \models \langle \! [\mathsf{C}] \! \rangle \varphi \text{ iff } \exists \overline{\beta_\mathsf{C}}\forall \overline{\beta_{\mathsf{S \setminus C}}}: 
        M,a \models \langle \overline{\sigma_{\mathsf{C}}}: \overline{\beta_{\mathsf{C}}} \rangle \top \text{ and }\\ 
        & \qquad \quad 
        M,a \models [{\overline{\sigma_\mathsf{C}}\cup \overline{\sigma_{\mathsf{S \setminus C}}}:\overline{\beta_\mathsf{C}}\cup \overline{\beta_{\mathsf{S \setminus C}}}} ] \varphi\\
        &M,a \models [ \! \langle\mathsf{C}\rangle \! ] \varphi \text{ iff } \forall \overline{\beta_\mathsf{C}}\exists \overline{\beta_{\mathsf{S \setminus C}}}: 
        M,a \models \langle \overline{\sigma_{\mathsf{C}}}: \overline{\beta_{\mathsf{C}}} \rangle \top \text{ implies }\\ 
        & \qquad \quad M,a \models \langle {\overline{\sigma_\mathsf{C}}\cup \overline{\sigma_{\mathsf{S \setminus C}}}:\overline{\beta_\mathsf{C}}\cup \overline{\beta_{\mathsf{S \setminus C}}}} \rangle \varphi
\end{alignat*}     
\end{definition}

Intuitively, $\langle \! [\mathsf{C}] \! \rangle \varphi$ holds if and only if there is a choice of buyers for the coalition of sellers $\mathsf{C}$ such that this choice is possible (part $\langle \overline{\sigma_{\mathsf{C}}}: \overline{\beta_{\mathsf{C}}} \rangle \top$) and whatever buyers the rest of the sellers decides to incentivise, $\varphi$ holds after the resulting joint concurrent action. 

It is immediate 
that the following formulas are valid\footnote{These formulas are some of the validites of CL \cite{pauly02}.}:
\begin{itemize}
    \item $\langle \! [\mathsf{C}] \! \rangle \varphi \to \langle \! [\mathsf{C} \cup \mathsf{D}] \! \rangle \varphi$, i.e., a superset of a coalition is at least as powerful as the coalition.
    \item $[ \! \langle\mathsf{\emptyset}\rangle \! ] \varphi \to \langle \! [\mathsf{S}] \! \rangle \varphi$, i.e., the relationship between the empty and the grand coalitions.
    \item $\langle \! [\mathsf{C}] \! \rangle (\varphi \land \psi) \to \langle \! [\mathsf{C}] \! \rangle \varphi$, i.e., the ability to achieve two goals implies the ability to achieve any single one of them.
\end{itemize}

\begin{example}
    The ability to reason about strategies of coalitions of sellers allows us to consider truly competitive and cooperative scenarios. For the first one, consider $\langle \! [\sigma_1,\sigma_2] \! \rangle [ \! \langle\sigma_3\rangle \! ] (ut_{\sigma_1} > ut_{\sigma_3})$ meaning that a coalition of the first two sellers 
    can preclude the third seller from having a utility equal to or higher than that of the first seller in a two-step incentivisation scenario. In a more altruistic setting, consider $\lnot \langle \! [\sigma_1,\sigma_2] \! \rangle (ut_{\sigma_1} + ut_{\sigma_2}>3) \land   \langle \! [\sigma_1,\sigma_2,\sigma_3] \! \rangle (ut_{\sigma_1} + ut_{\sigma_2}>3)$ meaning that together, the first and the second sellers cannot achieve a joint utility higher than 3, but if they cooperate with the third seller, this goal is satisfied.
\end{example}

\subsection{Expressivity and Model Checking}

Having introduced a strategic extension of $\mathcal{L}^n$, it is quite natural to wonder whether we gain anything in terms of expressivity, and, if yes, whether it comes at a price. We show that the answer to both questions is \textit{yes} (with a caveat). 

\begin{theorem}
\label{thm:truthpreserve}
Let $M,a$ be a finite $n$-DAM and $\varphi \in \mathcal{SL}^n$. Then there exists a $\psi \in \mathcal{L}^n$ s.t. $M,a \models \varphi$ iff $M,a \models \psi$.
\end{theorem}

\begin{proof}
    To prove the theorem, we present a truth-preserving translation $t:\mathcal{SL}^n \to \mathcal{L}^n$. All cases, apart from the strategic one, are trivial as $\mathcal{L}^n \subset \mathcal{SL}^n$. For the strategic case, we have $t(\langle \! [\mathsf{C}] \! \rangle \varphi) = \bigvee_{\overline{\beta_\mathsf{C}} \in \Nfin(B)^{|\mathsf{C}|}} \bigwedge_{\overline{\beta_{\mathsf{S \setminus C}}} \in \Nfin(B)^{|\mathsf{S \setminus C}|}} (\langle \overline{\sigma_{\mathsf{C}}}: \overline{\beta_{\mathsf{C}}} \rangle \top$ $\land$ 
$[{\overline{\sigma_\mathsf{C}}\cup \overline{\sigma_{\mathsf{S \setminus C}}}:\overline{\beta_\mathsf{C}}\cup \overline{\beta_{\mathsf{S \setminus C}}}} ] t(\varphi))$. Note that this is a well-formed formula because we are dealing with a finite mechanism and hence we can explicitly go over the elements in $\Nfin(B)^{|\mathsf{C}|}$ one by one, where $\Nfin(B)^{|\mathsf{C}|}$ is the set of all tuples of agent names of size $|\mathsf{C}|$. It follows from the definition of the semantics that the translation is truth-preserving and terminating.
\end{proof}

While Theorem \ref{thm:truthpreserve} presents a translation that, for a given finite mechanism $M,a$ and a formula $\varphi \in \mathcal{SL}^n$, produces a corresponding formula $\psi \in \mathcal{L}^n$ that agrees with $\varphi$ on $M,a$, this result cannot be extended to arbitrary DAMs. In particular, the next theorem states that it is not the case that for a given $\varphi \in \mathcal{SL}^n$ we can always find a $\psi \in \mathcal{L}^n$ that will agree with $\varphi$ on \textit{all} mechanisms. In other words, $\mathcal{SL}^n$ is more expressive than $\mathcal{L}^n$. To show this, observe that modalities $\langle \! [ \mathsf{C} ] \! \rangle \varphi$ quantify over \textit{all} buyers' names, even those that are not explicitly present in the formula. A sketch of the proof of the next theorem is given in the Technical Appendix. 

 
 
\begin{restatable}{theorem}{slexp}
    $\mathcal{SL}^n$ is strictly more expressive than $\mathcal{L}^n$, i.e. $\mathcal{L}^n \subset \mathcal{SL}^n$ and there is a $\varphi \in \mathcal{SL}^n$ s.t. for all $\psi \in \mathcal{L}^n$ there is a mechanism $M,a$ such that $M,a \models \varphi$ iff $M,a \not \models \psi$. 
\end{restatable}

As promised, the greater expressive power comes with a higher model checking complexity.


\begin{restatable}{theorem}{mcsl}

    Model checking 
    $\mathcal{SL}^n$ is PSPACE-complete for the class of finite DAMs with the placement, payment, and utility functions computable in polynomial space. 
\end{restatable}

\begin{proof}
To show that the problem is in PSPACE, we present Algorithm \ref{quantMC} that extends the P-time Algorithm \ref{diffMC}.
\begin{breakablealgorithm}
	\caption{An algorithm for model checking $\mathcal{SL}^n$}\label{quantMC} 
	\small
	\begin{algorithmic}[1] 		
		\Procedure{MC}{$M, a, \varphi$}		
\Case {$\varphi = \langle\! [ \mathsf{C} ] \!\rangle \psi$} 
\For{$b_1, ..., b_{|\mathsf{C}|} \in B$}
\State{flag: = \textit{true}}
\State{$\overline{\mathsf{\beta_{C}}} = \Nfin (b_1) ... \Nfin(b_{|\mathsf{C}|})$}
\If{\textsc{MC}($M, a, \langle {\overline{\sigma_\mathsf{C}} : \overline{\beta_\mathsf{C}}}\rangle\top$)}
    \For{$b_1, ..., b_{|\mathsf{S\setminus C}|} \in B$}
    \State{$\overline{\beta_{\mathsf{S \setminus C}}} = \Nfin (b_1) ... \Nfin(b_{|\mathsf{S\setminus C}|})$}
        \If{not \textsc{MC}($M, a, [{\overline{\sigma_\mathsf{C}} \cup \overline{\sigma_\mathsf{S \setminus C}} : \overline{\beta_\mathsf{C}} \cup \overline{\beta_\mathsf{{S \setminus C}}}}]\psi$)}
        \State{flag: = \textit{false}}
        \State{\textbf{break}}
        \EndIf
    \EndFor
\Else
\State{flag:= \textit{false}}
\EndIf
\If{flag}
\State{\textbf{return} \textit{true}}
\EndIf
\EndFor
\State{\textbf{return} \textit{false}}

\EndCase
   \EndProcedure

	\end{algorithmic}
\end{breakablealgorithm}

The only new case is $\varphi = \langle \! [ \mathsf{C} ] \! \rangle \psi$. The treatment of the case follows the semantics and thus the correctness follows. Indeed, to verify whether for a given agent 
we have $\langle \! [ \mathsf{C} ] \! \rangle \psi$, we look for a set of buyers that sellers in $\mathsf{C}$ will incentivise (lines 3--6), and then 
we check whether for all possible incentive diffusions by the remaining sellers (lines 7--8) we still can satisfy $\psi$ after each seller performs their action (line 9). Variable `flag' 
keeps 
track of whether we have found such a choice for the sellers in $\mathsf{C}$, and the algorithm returns $true$ is yes, and $false$ if not. 

Complexity analysis (sketch): observe that the algorithm checks an exponential number of subsets of buyers, and hence the running time is exponential in the size of the input. However, the algorithm uses only a polynomial amount of space. 
Note that we explore the tree of mechanism updates in a depth-first manner.  
The space required by a branch in such a tree, and hence by the algorithm, is bounded by $O(|\varphi|\cdot |M|^2)$. PSPACE-hardness is shown via a reduction from the QBF problem (see the Technical Appendix).
\end{proof}

\section*{Related Work}




\paragraph{Logics for auctions} 
Logic-based formalisms have been developed to capture and reason about various aspects of auctions. One direction is  \textit{bidding languages}, 
 most notably OR and XOR-based languages, which express the preferences of auction participants (see \cite{Nisan2000} for an overview).
These languages compactly represent possible bids on item combinations, an important aspect in combinatorial auctions, where bidders place bids on bundles of distinct items. 
Our work, instead, is closer to logical approaches for reasoning about and designing auctions. \citet{MittelmannJAAMAS22} proposed a lightweight formalism to represent auction rules, whereas \citet{MittelmannAAMAS22} addressed the representation of strategies in repeated auctions, and \citet{MittelmannHP21} captured bounded rationality in auctions.   
Another line of research proposes the use of variants of Strategy Logic (SL) for the design of auction mechanisms, exploiting verification and synthesis 
\cite{MittelmannAIJ2025,MittelmannMMP23}. Notably, the model-checking complexity for specifications in SL is non-elementary in the general case \cite{MogaveroMPV14}. 
There is also research on automated verification of auction protocols \cite{GargRSS25,CaminatiK0R15,CKMRWW13} that stresses the importance of formal verification techniques for auctions. 

Across all these works, the set of agents involved in the auction 
is fixed throughout its execution. 
To the best of our knowledge, our work is the first one to explore the dynamics of auction diffusion from a logic-based perspective. 


\paragraph{DEL and social network logics}
Our intuition of model updates stems from \textit{dynamic epistemic logic} (DEL), 
where one can model various information-changing events 
in the context of agents' knowledge. Ideas of DEL were also adopted in the field of \textit{social network logics} (SNLs), where one uses formal tools to study such phenomena on social networks as information diffusion \cite{christoff15,BaltagCRS19}, social influence \cite{ChristoffHP16}, and echo chambers \cite{PedersenSA19}, to name a few. See \cite[Chapter 3]{minathesis} for an overview. Perhaps the most related work here is \cite{galimullinP24}, where the authors explore visibility of posts on social networks, and how these posts propagate through the network, somewhat akin to how the information about an auction is spread in diffusion auctions. Using nominals for agent names is common in SNLs and comes from hybrid logic (see, e.g., \cite{ARECES2007821}). While discussing strategic logic $\mathcal{SL}^n$, we noted that our coalitional operators are inspired by those of CL and ATL. However, in CL and ATL models are static, i.e., they do not change as a result of agents' actions. Hence, a more relevant work is that on \textit{coalition announcements} \cite{agotnes08,galimullin21b,delima14} in DEL, where strategic operators quantify over model changes that agents can bring about in a competitive setting, and with model checking complexity being PSPACE-complete \cite{AlechinaDGW21}. 
Another related work is \cite{MaubertPSS20}, where agents play a multi-step concurrent game by modifying a model using modalities of DEL. Finally, there has been some work on adding arrows in modal logics 
\cite{ArecesFH15}.

\section{Conclusion}
We have presented a formal framework for reasoning about sellers' strategies in diffusion auctions. In particular, we introduced two logics, the $n$-seller logic for diffusion incentives $\mathcal{L}^n$ and its strategic version $\mathcal{SL}^n$, that can capture various properties of such auctions, like item allocations, utility increase, local properties of the underlying social network, and Nash equilibrium, to name a few. Our logics are \textit{dynamic}, and hence they also allow us to verify whether the above-mentioned properties hold after modifications of the underlying social network that are engendered by sellers incentivising buyers to invite their friends to join auctions. To the best of our knowledge, this is \textit{the first work that tackles the problem of formal verification of diffusion auctions}.  

Our definition of diffusion auction mechanisms is quite general and allows us to capture a variety of auction types as long as the complexity of computing the placement, payment, and utility functions is no higher than the complexity of the model-checking problem of the corresponding logic. 
We have shown that it is in P for $\mathcal{L}^n$ and is PSPACE-complete for $\mathcal{SL}^n$. Moreover, we have demonstrated that the complexity of the strategy existence problem for a given mechanism and a joint goal of sellers is NP-complete.   

With our work, we start a research line on formal verification of diffusion auctions, and there are plenty of interesting further directions. 
In particular, we would like to tackle the formal verification of 
a probabilistic framework, capturing incomplete information and Bayesian analysis
\cite{huang2025approximate}, as well as consider strategies of buyers. We have also mentioned some of the validities of our logics, and we find it very tempting to explore their axiomatisations. Finally, we plan to explore the case of multi-item diffusion auctions.

\section*{Acknowledgements}
This research is partially supported by the ANR project NOGGINS ANR-24-CE23-4402.

\bibliography{ref}

\clearpage
\appendix
\section{Technical Appendix}
\strex*
\begin{proof}
   That the strategy existence problem is in NP follows from the fact that model checking $\mathcal{L}^n$ is in P. Indeed, since the mechanism is finite and sellers have a finite budget, there is an upper bound on the length of $\langle \overline{\sigma}: \overline{\beta} \rangle^\ast$. In particular, the length of $\langle \overline{\sigma}: \overline{\beta} \rangle^\ast$ is polynomial in the size of $M$ (in the worst case, each seller incentivises each buyer and the total number of incentivisations is bounded by sellers' budgets that are a part of the input). Therefore, we can guess a sequence $\langle \overline{\sigma}: \overline{\beta} \rangle^\ast$ of at most polynomial size, and then verify $M,s \models \langle \overline{\sigma}: \overline{\beta} \rangle^\ast \varphi$ for sellers $s \in S$ in polynomial time.

    For NP-hardness, it is enough to show only the case of the single seller mechanism. To this end, we employ a reduction from 3-SAT. Let $\Psi = \bigwedge_{1 \leqslant i \leqslant k} \beta_i$ be a 3-SAT instance with $\beta_i = \gamma^i_1 \lor \gamma^i_2 \lor \gamma^i_3$, where $\gamma^i$ are literals. Moreover, let $\Delta = \{\delta_1,..., \delta_n\}$ be the set of atoms occurring in $\Psi$.

    Given a $\Psi$, we will construct a 1-DAM $M_\Psi$ with arbitrary polynomially computable $P$, $\Pay$, $Ut$. The set of agents $Agt = \{s, b_1,...,b_k,$ $c^1_{1,1}, c^1_{2,1}, c^1_{3, 1}, ....,$ $c^k_{3, n}, d_1,...,$ $d_n, e^1_1, e^1_2, ..., e^n_1, e^n_2, t_1, ..., t_n, f_1, ..., f_n\}$, where $s$ is the seller, agents $b_i$ correspond to clauses in $\Psi$, agents $c^i_{j,l}$ correspond to literals in a given clause $i$ with $j$ being the position of the literal in the clause and $l$ the number of the corresponding atom from $\Delta$, agents $d_i$ correspond to atoms in $\Delta$, agents $e^i_j$ help to split the valuations of atoms, and $t_i$ and $f_i$ correspond to setting the corresponding atom to \textit{true} or \textit{false}. Relation $F$ is set up in a hierarchical way, with $F(s, \{b_i | 1 \leqslant i \leqslant k\})$, $F(b_i, \{c^i_{1,l}, c^i_{2,l'}, c^i_{3,l''}\})$ for all $b_i$, $F(c^i_{j,l}, d_l)$ for all $c^i_{j,l}$ such that $d_l$ is the corresponding atom for literal $c^i_{j,l}$, $F(d_i, \{e^i_1, e^i_2\})$ for all $d_i$, and $F(e^i_1, t_i)$ and $F(e^i_2, f_i)$ for all $e^i_j$. Further, $Bdg$ and $V$ are defined arbitrarily, and $I(a) = 0$ for all $a\in Agt$. Finally, $N(\sigma) = s$, $N(\beta_i) = b_i$, $N(\gamma^i_{j,l}) = c^i_{j,l}$ if the corresponding literal is positive and $N(\overline{\gamma^i_j}) = c^i_j$ otherwise, $N(\delta_i) = d_i$, $N(\epsilon^i_j) = e^i_j$,  $N(true_i) = t_i$, and $N(false_i) = f_i$. All other names are distributed arbitrarily\footnote{The reader might have noticed that we use the same symbols for atoms in $\Psi$ and names of $d$-agents. This is done on purpose to make Figure \ref{fig:satexample} more readable and to highlight that $d$-agents correspond to atoms. Similarly for clauses and literals.}.

    As an example of the construction, consider the 3-SAT instance $\Psi = (\delta_1 \lor \delta_2 \lor \delta_3) \land (\lnot \delta_1 \lor \delta_3 \lor \lnot \delta_4)$, and the corresponding mechanism $M_\Psi$ in Figure \ref{fig:satexample}.  

    \begin{figure}[h!]
\centering
\scalebox{0.8}{
   \begin{tikzpicture}[rotate = 90]
\node[circle,draw=black, minimum size=4pt,inner sep=0pt, fill = black, label=above right:{$s:\sigma$}](s) at (0,0) {};
\node[circle,draw=black, minimum size=4pt,inner sep=0pt,, fill = black , label=left:{$b_1:\beta_1$}](b1) at (-4.5,-1.5) {};
\node[circle,draw=black, minimum size=4pt,inner sep=0pt, , fill = black, label=left:{$b_2:\beta_2$}](b2) at (4.5,-1.5) {};
\node[circle,draw=black, minimum size=4pt,inner sep=0pt,, fill = black , label=left:{$c_{1,1}^1:\gamma_{1,1}^1$}](c11) at (-6.5,-3) {};
\node[circle,draw=black, minimum size=4pt,inner sep=0pt,, fill = black , label=above:{$c_{2,2}^1:\gamma_{2,2}^1$}](c12) at (-4.5,-3) {};
\node[circle,draw=black, minimum size=4pt,inner sep=0pt,, fill = black , label=left:{$c_{3,3}^1:\gamma_{3,3}^1$}](c13) at (-2.5,-3) {};

\node[circle,draw=black, minimum size=4pt,inner sep=0pt,, fill = black , label=right:{$c_{3,4}^2:\overline{\gamma_{3,4}^2}$}](c23) at (6.5,-3) {};
\node[circle,draw=black, minimum size=4pt,inner sep=0pt,, fill = black , label=above:{$c_{2,3}^2:\gamma_{2,3}^2$}](c22) at (4.5,-3) {};
\node[circle,draw=black, minimum size=4pt,inner sep=0pt,, fill = black , label=left:{$c_{1,1}^2:\overline{\gamma_{1,1}^2}$}](c21) at (2.5,-3) {};

\node[circle,draw=black, minimum size=4pt,inner sep=0pt,, fill = black , label=below:{$d_2:\delta_2$}](d2) at (-5,-5) {};
\node[circle,draw=black, minimum size=4pt,inner sep=0pt,, fill = black , label=left:{$d_1:\delta_1$}](d1) at (-2,-5) {};
\node[circle,draw=black, minimum size=4pt,inner sep=0pt,, fill = black , label=left:{$d_3:\delta_3$}](d3) at (2,-5) {};
\node[circle,draw=black, minimum size=4pt,inner sep=0pt,, fill = black , label=below:{$d_4:\delta_4$}](d4) at (5,-5) {};

\node[circle,draw=black, minimum size=4pt,inner sep=0pt,, fill = black , label=below:{$e^2_1$}](e21) at (-6,-6.5) {};
\node[circle,draw=black, minimum size=4pt,inner sep=0pt,, fill = black , label=below:{$e^2_2$}](e22) at (-4,-6.5) {};
\node[circle,draw=black, minimum size=4pt,inner sep=0pt,, fill = black , label=above:{$e^1_1$}](e11) at (-3,-6.5) {};
\node[circle,draw=black, minimum size=4pt,inner sep=0pt,, fill = black , label=below:{$e^1_2$}](e12) at (-1,-6.5) {};
\node[circle,draw=black, minimum size=4pt,inner sep=0pt,, fill = black , label=below:{$e^3_1$}](e31) at (1,-6.5) {};
\node[circle,draw=black, minimum size=4pt,inner sep=0pt,, fill = black , label=below:{$e^3_2$}](e32) at (3,-6.5) {};
\node[circle,draw=black, minimum size=4pt,inner sep=0pt,, fill = black , label=below:{$e^4_1$}](e41) at (4,-6.5) {};
\node[circle,draw=black, minimum size=4pt,inner sep=0pt,, fill = black , label=above:{$e^4_2$}](e42) at (6,-6.5) {};

\node[circle,draw=black, minimum size=4pt,inner sep=0pt,, fill = black , label=below:{$t_2:true_2$}](t2) at (-6,-8) {};
\node[circle,draw=black, minimum size=4pt,inner sep=0pt,, fill = black , label=below:{$f_2:false_2$}](f2) at (-4,-8) {};
\node[circle,draw=black, minimum size=4pt,inner sep=0pt,, fill = black , label=above:{$t_1:true_1$}](t1) at (-3,-8) {};
\node[circle,draw=black, minimum size=4pt,inner sep=0pt,, fill = black , label=below:{$f_1:false_1$}](f1) at (-1,-8) {};
\node[circle,draw=black, minimum size=4pt,inner sep=0pt,, fill = black , label=above:{$t_3:true_3$}](t3) at (1,-8) {};
\node[circle,draw=black, minimum size=4pt,inner sep=0pt,, fill = black , label=below:{$f_3:false_3$}](f3) at (3,-8) {};
\node[circle,draw=black, minimum size=4pt,inner sep=0pt,, fill = black , label=above:{$t_4:true_4$}](t4) at (4,-8) {};
\node[circle,draw=black, minimum size=4pt,inner sep=0pt,, fill = black , label=below:{$f_4:false_4$}](f4) at (6,-8) {};

\draw[thick] (s) to (b1);
\draw[thick] (s) to (b2);

\draw[thick] (c11) to (b1);
\draw[thick] (c12) to (b1);
\draw[thick] (c13) to (b1);

\draw[thick] (c21) to (b2);
\draw[thick] (c22) to (b2);
\draw[thick] (c23) to (b2);

\draw[thick] (c11) to (d1);
\draw[thick] (c21) to (d1);
\draw[thick] (c12) to (d2);
\draw[thick] (c13) to (d3);
\draw[thick] (c22) to (d3);
\draw[thick] (c23) to (d4);

\draw[thick] (e11) to (d1);
\draw[thick] (e12) to (d1);
\draw[thick] (e21) to (d2);
\draw[thick] (e22) to (d2);
\draw[thick] (e31) to (d3);
\draw[thick] (e32) to (d3);
\draw[thick] (e41) to (d4);
\draw[thick] (e42) to (d4);

\draw[thick] (e11) to (t1);
\draw[thick] (e12) to (f1);
\draw[thick] (e21) to (t2);
\draw[thick] (e22) to (f2);
\draw[thick] (e31) to (t3);
\draw[thick] (e32) to (f3);
\draw[thick] (e41) to (t4);
\draw[thick] (e42) to (f4);

\draw[thick,dashed, bend left] (t3) to (s);

\end{tikzpicture}
 }
\caption{Mechanism $M$ for the 3-SAT instance $\Psi$ with one new link being dashed, and other new links being omitted for readability.}
\label{fig:satexample}
\end{figure}
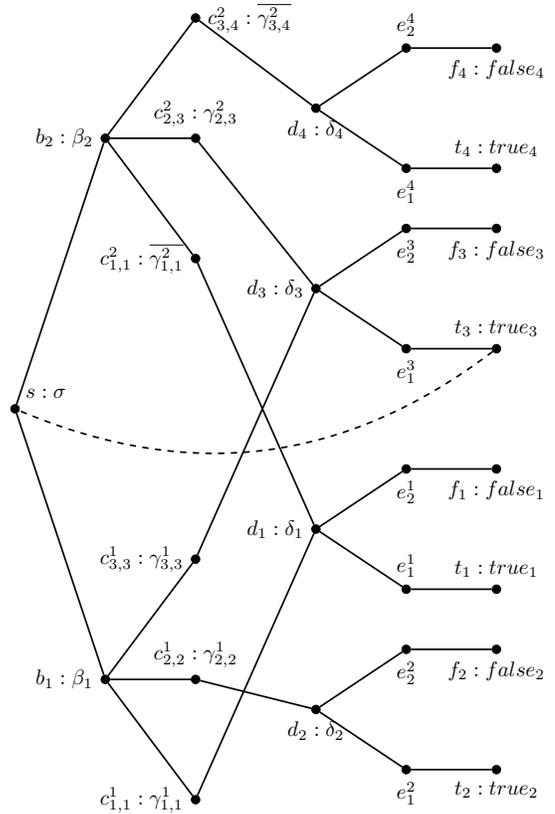 

As for goal formula $\varphi_\Psi$, we take $$\varphi_\Psi = \square 
\left( \beta_i \to \Diamond 
\left((\gamma^i_{j,l} \lor \overline{\gamma^i_{j,l}}) \land \psi_0 \land \psi_1\right)\right)$$
$\text{ for } 1\leqslant i \leqslant k, 1 \leqslant j \leqslant 3 \text { and } 1 \leqslant l \leqslant n$,
where $$\psi_0 = \overline{\gamma^i_{j,l}} \to (\Diamond \Diamond \Diamond (
false_l \land \Diamond \sigma) \land \lnot \Diamond \Diamond \Diamond (
true_l \land \Diamond \sigma))$$ $$\psi_1 = \gamma^i_{j,l} \to (\Diamond \Diamond \Diamond (
true_l \land \Diamond \sigma) \land  \lnot \Diamond \Diamond \Diamond (
false_l \land \Diamond \sigma)).$$

The goal formula is constructed in such a way in order to mimic, together with the model $M_\Psi$, the instance of 3-SAT $\Psi$. Hence, it is easy to verify that there is a sequence of incentives for the seller such that $M_\Psi, s \models \langle \alpha \rangle^\ast \varphi_\Psi$ if and only if $\Psi$ is satisfiable. Indeed, to satisfy $\Psi$ we need to satisfy all its clauses, which, in the goal formula, corresponds to the fact that we visit all friends of $s$ with names $\beta_i$. Then, for each clause, we need to find at least one literal, positive or negative, which is true. In the goal formula, this corresponds to checking the existence of a at least one friend (of three total) of a given $\beta_i$ which is a literal (part $\gamma^i_{j,l} \lor \overline{\gamma^i_{j,l}}$), and if the literal is negative (case $\psi_0$), then the corresponding \textit{false} agent is reachable that participates in the auction, and no corresponding \textit{true} agent is reachable such that the agent is friends with the seller (i.e. that the truth value of the corresponding atom is set only to \textit{false}). Similarly for $\psi_1$. 

In other words, if $\Psi$ is satisfiable, then an atom $\delta_i$ in $\Psi$ is true if and only if the corresponding $t_i$ agent participates in the auction. And vice versa, if there is a sequence of incentives that satisfies the goal formula, then we set an atom $\delta_i$ in $\Psi$ to true if and only if the corresponding $t_i$ agent is participating in the auction.  

For the given example of $\Psi$, the sellers strategy could be $\langle \beta_1\rangle \langle \gamma^1_{3,3} \rangle \langle \delta_3 \rangle \langle \epsilon_1^3 \rangle$, and the result of executing such a strategy is presented in Figure \ref{fig:satexample} (with one dashed new line, all other new lines are omitted for readability). It is easy to verify that in the updated mechanism, the goal formula holds, which, for the given example, corresponds to setting atom $\delta_3$ to true and thus satisfying $\Psi$.
\end{proof}

\slexp*
\begin{proof}
         We need to show that there is a formula of $\mathcal{SL}^n$ for which there is no equiavalent formula of $\mathcal{L}^n$. Consider $\langle \! [ \sigma ] \! \rangle \Diamond \gamma \in \mathcal{SL}^1$ constructed over the single seller, and assume towards a contradiction that there is an equivalent formula $\varphi \in \mathcal{L}^1$. Recall that $name(\varphi)$ be the set of nominals appearing in $\varphi$. The set is finite since $\varphi$ is finite. 
     
     Now consider the following two 1-DAMs depicted in Figure \ref{fig:exprproof}, where $\alpha_1,...,\alpha_n$ are all buyer names from $name(\varphi) \setminus \{\gamma\}$.
     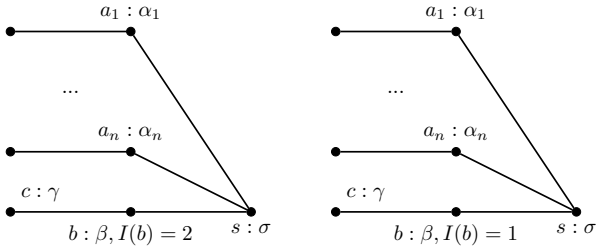
\begin{figure}[h!]
\centering
\scalebox{0.8}{
   \begin{tikzpicture}[rotate = -90]
\node[circle,draw=black, minimum size=4pt,inner sep=0pt, fill = black, label=below:{$s:\sigma$}](s) at (0,0) {};
\node[circle,draw=black, minimum size=4pt,inner sep=0pt,, fill = black , label=above:{$a_1:\alpha_1$}](a10) at (-3,-2) {};
\node[circle,draw=black, minimum size=4pt,inner sep=0pt, , fill = black, label=below:{$b:\beta, I(b) = 2$}](a21) at (0,-2) {};
\node[circle,draw=black, minimum size=4pt,inner sep=0pt, , fill = black, label=above:{$a_n:\alpha_n$}](a20) at (-1,-2) {};

\node[circle,draw=black, minimum size=4pt,inner sep=0pt,, fill = black ](b10) at (-3,-4) {};
\node[circle,draw=black, minimum size=4pt,inner sep=0pt, , fill = black, label=above right:{$c:\gamma$}](b21) at (0,-4) {};
\node[circle,draw=black, minimum size=4pt,inner sep=0pt, , fill = black](b20) at (-1,-4) {};

\node (dots) at (-2,-3) {$...$};

\draw[thick] (s) to (a10);
\draw[thick] (s) to (a20);
\draw[thick] (s) to (a21);

\draw[thick] (b10) to (a10);
\draw[thick] (b20) to (a20);
\draw[thick] (b21) to (a21);

\end{tikzpicture}
\hspace{7mm}
   \begin{tikzpicture}[rotate = -90]
\node[circle,draw=black, minimum size=4pt,inner sep=0pt, fill = black, label=below:{$s:\sigma$}](s) at (0,0) {};
\node[circle,draw=black, minimum size=4pt,inner sep=0pt,, fill = black , label=above:{$a_1:\alpha_1$}](a10) at (-3,-2) {};
\node[circle,draw=black, minimum size=4pt,inner sep=0pt, , fill = black, label=below:{$b:\beta, I(b) = 1$}](a21) at (0,-2) {};
\node[circle,draw=black, minimum size=4pt,inner sep=0pt, , fill = black, label=above:{$a_n:\alpha_n$}](a20) at (-1,-2) {};

\node[circle,draw=black, minimum size=4pt,inner sep=0pt,, fill = black ](b10) at (-3,-4) {};
\node[circle,draw=black, minimum size=4pt,inner sep=0pt, , fill = black, label=above right:{$c:\gamma$}](b21) at (0,-4) {};
\node[circle,draw=black, minimum size=4pt,inner sep=0pt, , fill = black](b20) at (-1,-4) {};

\node (dots) at (-2,-3) {$...$};

\draw[thick] (s) to (a10);
\draw[thick] (s) to (a20);
\draw[thick] (s) to (a21);

\draw[thick] (b10) to (a10);
\draw[thick] (b20) to (a20);
\draw[thick] (b21) to (a21);

\end{tikzpicture}

 }

\caption{Mechanisms $M_1$ (left) and $M_2$ (right).}
\label{fig:exprproof}
\end{figure} 

In the models, the seller has budget 1 and the incentives required by each her neighbour $a_i$ is 1, i.e. the seller can incentivise only one buyer. Budgets of all buyers are 0. Buyer $b$, whose name $\beta$ does not appear in $names (\varphi)$ (we can assume the existence of such a name as the set of nominals is countably infinite) has incentive 2 in $M_1$ and 1 in the model $M_2$. It is easy to see that $M_1,s \not \models \langle \! [ \sigma ] \! \rangle \Diamond \gamma$, as the only way for the seller to reach buyer $\gamma$ is by incentivising $\beta$ to invite them. Since the seller's budget is 1, they do not have enough resources to do so. At the same time, $M_2,s  \models \langle \! [ \sigma ] \! \rangle \Diamond \gamma$. Again, even though $\beta$ does not appear explicitly in neither $\varphi$ nor in $\langle \! [ \sigma ] \! \rangle \Diamond \gamma$, modality $\langle \! [ \sigma ] \! \rangle$ still quantifies over this name.

To see that $M_1,s\models \varphi$ if and only if $M_2,s \models \varphi$, it is enough to notice that the distributions of names in mechanisms are identical, and, moreover, since $\beta$ does not appear in $\varphi$, $\beta$ is not a part of the dynamic operators. Finally, we cannot spot the difference in the incentives for agent $b$ in the two mechanism without using incentivisations, which is ruled out due to $\beta$ not being a part of $\varphi$. 
\end{proof}

\mcsl*
\begin{proof}
      For PSPACE-hardness, we employ the reduction from the quantified Boolean formula (QBF) problem, which is known to be PSPACE-complete. The problem consists in determining whether a given QBF $\Psi = Q_1 p_1 ... Q_n p_n \psi(p_1,...,p_n)$ with $Q_i \in \{\forall, \exists\}$ is true. Without loss of generality, we assume that there are no free variables in $\Psi$. For the reduction, we will construct a mechanism $M_\Psi$ over one seller and a formula $\varphi_\Psi$ such that QBF $\Psi$ is true if and only if $M_\Psi, s \models \varphi_\Psi$. 

    Given a QBF $\Psi = Q_1 p_1 ... Q_n p_n \psi(p_1,...,p_n)$, mechanism $M_\Psi$ is the tuple with $Agt = \{s, a_1^0,a_1^1, ..., a_n^0, a_n^1, b_1^0, b_1^1, ..., b_n^0, b_n^1\}$ with $s$ being the seller. Relation $F$ is such that $F(s, a_i^j)$ for all $a_i^j$ and $F(a_i^j,b_i^j)$ for all $a_i^j$ and $b_i^j$. $Bdg$ and $V$ are defined arbitrarily, and $I(a,s) = 0$ for all $a \in Agt$. Functions $N(\sigma) = s$, $N(\alpha_i^j) = a^j_i$,  and $N(\beta_i^j) = b^j_i$. All other names are distributed arbitrarily. Finally, the placement, payment, and utility functions are arbitrary. 
    
    Intuitively, mechanism $M_\Psi$ consists of $2n+1$ agents, where agents $a^0_i$ and $b^0_i$ model the setting atom $p_i$ to \textit{false}, and agents $a^1_i$ and $b^1_i$ model the setting atom $p_i$ to \textit{true}. Which truth value of atom $p_i$ is chosen, \textit{true} or \textit{false}, will be determined whether the corresponding buyer is participating in the auction run by the seller $s$.

    We use the following translation of QBF $\Psi$ into a formula of our logic. First, formula $fixed_k$ intuitively means that truth values of first $k$ atoms in $\Psi$ were chosen:
    $$fixed_k = \bigwedge_{1 \leqslant i \leqslant k} (\Diamond \beta_i^0 \leftrightarrow \lnot \Diamond \beta_i^1) \land \bigwedge_{k < i \leqslant n} (\lnot \Diamond \beta^0_i \land \lnot \Diamond \beta^1_i).$$

    Now,
    \begin{align*}
    \varphi_0 &:= \psi(\Diamond \beta^1_1, ..., \Diamond \beta^1_n)\\
\varphi_k &:= 
\begin{cases} 
	[\! \langle \sigma \rangle \!] (fixed_k \to \varphi_{k-1}) &\text{if } Q_k = \forall\\
	\langle \! [ \sigma ] \! \rangle (fixed_k \land \varphi_{k-1}) &\text{if } Q_k = \exists\\
\end{cases}\\
\varphi_\Psi &:= \varphi_{n}.
\end{align*}

What is left to show is that $\Psi = Q_1 p_1 ... Q_n p_n \psi(p_1,...,p_n)$ is satisfiable if and only if $M_\Psi, s \models \varphi_\Psi$. For this, observe that $M_\Psi, s \models \Diamond \beta_i^1$, i.e. the buyer called $\beta^1_i$ is participating in the auction, holds if the the corresponding atom $p_i$ was set to \textit{true}. And similarly for $\beta_i^0$ and setting $p_i$ to \textit{false}. Since there is only a single seller $\sigma$, setting truth values of $p_i$'s happens sequentially one after another. Moreover, again since we have only the single seller, modalities $[\! \langle \sigma \rangle \!]$ (resp. $\langle \! [ \sigma ] \! \rangle$) correspond to universally (resp. existentially) quantifying over $p_i$. Guards $fixed_k$ guarantee that only truth values of only first $k$ atoms were chosen (conjunct $\bigwedge_{k < i \leqslant n} (\lnot \Diamond \beta^0_i \land \lnot \Diamond \beta^1_i)$), and, additionally, that those truth values were chosen unambiguously, i.e. exactly one truth value per atom (conjunct $\bigwedge_{1 \leqslant i \leqslant k} (\Diamond \beta_i^0 \leftrightarrow \lnot \Diamond \beta_i^1)$). Hence, together with $[\! \langle \sigma \rangle \!]$ and $\langle \! [ \sigma ] \! \rangle$, guards $fixed_k$ emulate quantifiers $\forall$ and $\exists$ in a QBF. Finally, the evaluation of the propositional part of the QBF is captured by the participation of the corresponding buyers in the auction. 
\end{proof}

\clearpage
\end{document}